\documentclass[a4paper,UKenglish,10pt]{llncs}
\usepackage[utf8]{inputenc}
\usepackage{amsmath}
\usepackage{amsfonts}
\usepackage{amssymb}
\usepackage{graphicx}
\usepackage{xcolor}
\usepackage[left=1.8in,right=1.8in,top=1.8in,bottom=1.8in]{geometry}
\usepackage{algpseudocode}
\usepackage{algorithm}

\PassOptionsToPackage{normalem}{ulem}
\usepackage{ulem}
\providecolor{added}{rgb}{0,0,1}
\providecolor{deleted}{rgb}{1,0,0}
\providecolor{mynote}{rgb}{0,0.5,0}

\date{October 1, 2018}

\title{Computing the Minkowski Sum of Convex Polytopes in $\Re^d$}

\author{Sandip Das\inst{1} 
	\and Swami Sarvottamananda\inst{2}}

\institute{Indian Statistical Institute, Kolkata (\email{sandipdas@isical.ac.in)} 
	\and Ramakrishna Mission Vivekananda Educational and Research Institute, Howrah (\email{sarvottamananda@rkmvu.ac.in)}}

\begin{document}%
\maketitle%
\begin{abstract}
We propose a method to efficiently compute the Minkowski sum, denoted by binary operator $\oplus$ in the paper,  of convex polytopes in $\Re^d$ using their face lattice structures as input.  In plane, the Minkowski sum of convex polygons can be computed in  linear time of the total number of vertices of the polygons. In $\Re^d$, we first show how to compute the Minkowski sum, $P \oplus Q$, of two convex polytopes $P$ and $Q$ of input size $n$ and $m$ respectively in time $O(nm)$. Then we generalize the method to compute the Minkowski sum of $n$ convex polytopes, $P_1 \oplus P_2 \oplus \cdots \oplus P_n$, in $\Re^d$ in time $O(\prod_{i}^{n}N_i)$, where $P_1$, $P_2$, $\dots$, $P_n$ are $n$ input convex polytopes and for each $i$, $N_i$ is size of the face lattice structure of $P_i$. Our algorithm for Minkowski sum of two convex  polytopes is optimal in the worst case since the output face lattice structure of $P\oplus Q$ for convex polytopes in $\Re^d$ can be $O(nm)$ in worst case.
\end{abstract}

\section{Introduction}
\label{sec:intro}

The Minkowski sum is an important concept, initially defined by Hermann Minkowski [1864-1909], in the field of computational geometry, although uncited. A formal definition is given later in Section~\ref{sec:defn}. The Minkowski sum,  in $\Re^2$ and $\Re^3$,  of polygons and polyhedra respectively, is used in computer graphics, robotic motion planning, computer-aided design and manufacturing. The abundant literature from late nineteenth century on Minkowski sum corroborates the importance and applicability of the concept.
Another recent use of Minkowski sums, that is not highlighted much in the literature, is to compute distances for several types of polyhedral distance functions~\cite{ChewD1985,DasNS2018}, especially when distances from polytopes are involved. Das et al.~\cite{DasNS2018} used Minkowski sums implicitly to compute diameter, width, minimum enclosing/stabbing sphere, maximum inscribed sphere and minimum enclosing/stabbing cylinder for several types of input and problem variations involving convex polygon/polytopes.

The computation of  Minkowski sum for polygons in $\Re^2$ is a well studied problem 
due to its importance and its occurance in several texts~\cite{PreparataS1985,BergOKO2008}. The algorithm to compute Minkowski sum of convex polygons in $\Re^2$ is simple and elegant, which we present in Section~\ref{sec:planeminkowski}. 
In $\Re^3$ too, in spite of increase in complexity, we have some studies done related to Minkowski sums of polyhedra. 
In $\Re^3$, Fogel et al.~\cite{FogelH2005} gave an $O(nm)$ algorithm to compute exact Minkowski sum of two convex polyhedra of $n$ and $m$ vertices respectively, which they claimed can be used to compute Minkowski sum of general polyhedra. In this paper, as a consequence of our results of computing Minkowski sum of convex polytopes in $\Re^d$, we provide an alternative method to compute the Minkowski sum of two or more convex polyhedra in $\Re^3$ using face lattice structures, formally defined e.g. by Edelsbrunner~\cite{Edelsbrunner1987}. The face lattice structures are used extensively for representing polytopes in $\Re^d$.

There are some methods in the literature to compute the Minkowski sum of non-convex polytopes, in the plane and in $\Re^3$.
Chazelle \cite{Chazelle1981} gave a method to decompose polyhedra in $\Re^3$ into a set of convex polyhedra. This concept of decomposition was used by Agarwal et al.~\cite{AgarwalFH2002} to compute Minkowski sum of simple polygons in plane efficiently, and is also effective in $\Re^3$ as claimed by Fogel et al.~\cite{FogelH2005}.

There are several studies on the upper bounds on the number of faces and size of the Minkowski sums of two convex polytopes and also on the upper bounds in the case of more than two convex polytopes. For two convex polytopes, it is well known that the number of features of the Minkowski sum is the product of number of features of two input convex polytopes in general case.
Gritzmann et al.~\cite{GritzmannS1993} studied the complexity of Minkowski sum of polytopes in terms of the constituent polytopes' features. Karavelas et al.~\cite{KaravelasMT2012} studied the complexity of Minkowski sum of three convex polytopes by considering $d$-cyclic polytopes $C_d(n)$ as input. 

In this paper, we use face lattice structures of convex polytopes to compute the Minkowski sum in $\Re^d$. Also, after the computation of the Minkowski sum, we output the face lattice structure of the resulting polytope. If the sizes of face lattice structures of the two input convex polytopes $P$ and $Q$ are $n$ and $m$ respectively, then our algorithm takes $O(nm)$ time to compute the face lattice structure of the polytope $P\oplus Q$. Furthermore we show that we can use the same techniques to compute the Minkowski sum of $N$ convex polytopes $P_1 \oplus P_2 \oplus \cdots \oplus P_N$ in $\Re^d$ in time $O(\prod_{i}^{n}N_i)$, where $P_1$, $P_2$, $\dots$, $P_N$ are $N$ convex polytopes in $\Re^d$ and for each $i$, $1\leq i\leq N$, $N_i$ is the size of the face lattice structure of the polytope $P_i$.

Generally, there are two types of approaches to compute Minkowski sums of convex polytopes. One is that we compute the correct extreme vertices of the final Minkowski sum and then we compute the convex hull of the computed extreme vertices. Second approach is to incrementally construct the polytope which represents the Minkowski sum by geometrical methods as we wrap around the input convex polytopes. The second method is known to work well in $\Re^2$ and $\Re^3$. However, generalizing the method to higher dimensions is not easy due to non-obvious nature of adjacencies and incidences.

In our face lattice method, we give up the approach of using adjacencies to wrap around the input convex polytopes. Instead we use the local information stored in each face node in face lattice to compute the Minkowski sum (the face nodes are nodes in a face lattice). This novel approach enables us to compute the Minkowski sum in the time complexity which is tight for worst case input. This is our contribution in this paper.
 
In Section~\ref{sec:defn} we give some definitions and notations that we use in the paper. 
We briefly present how the Minkowski sum of convex polygons in plane is computed, with respect to our face lattice method, in Section~\ref{sec:planeminkowski} to explain the concepts that we use later.
In Section~\ref{sec:genminkowski} we present the algorithm to compute the Minkowski sum of two polytopes in $\Re^d$. 
The correctness of the algorithm follows from the lemmas presented in the same section.
We generalize the method to efficiently compute the Minkowski sum of multiple polytopes in $\Re^d$  in Section~\ref{sec:multiple}. Finally, we give some concluding remarks in Section~\ref{sec:conclusion}.

\section{Definitions and Notations}
\label{sec:defn}
\begin{figure*}[t]
	\begin{minipage}[b]{0.485\textwidth}
		\centering %
		\includegraphics[width=0.9\columnwidth]{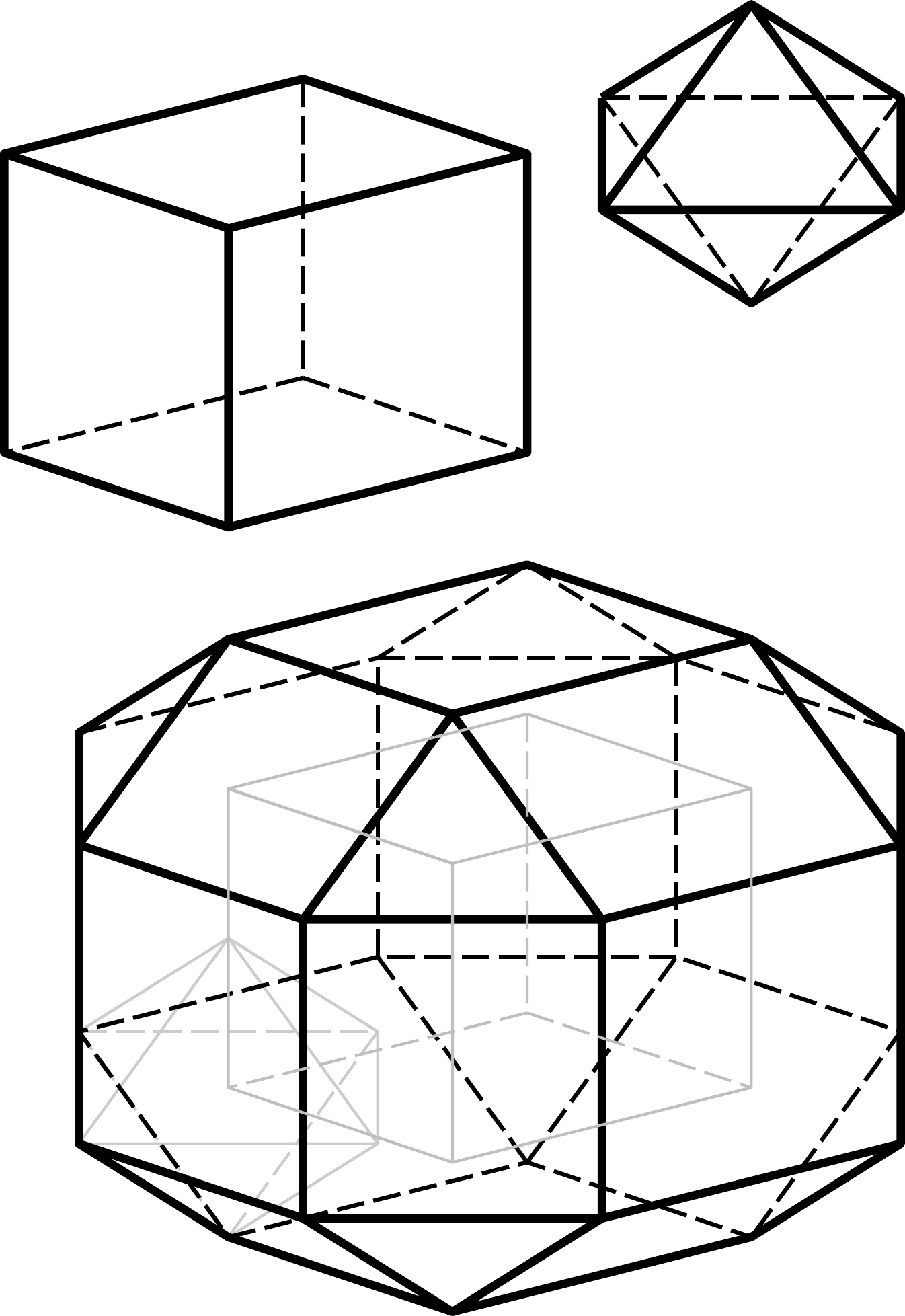} %
		\caption{Minkowski sum of two convex polytopes in $\Re^3$.%
		}%
		\vspace*{0ex}
		\label{fig:minkowski}
	\end{minipage}
	\hfill\quad%
	\begin{minipage}[b]{0.485\textwidth}
		\centering %
		\includegraphics[width=0.9\columnwidth]{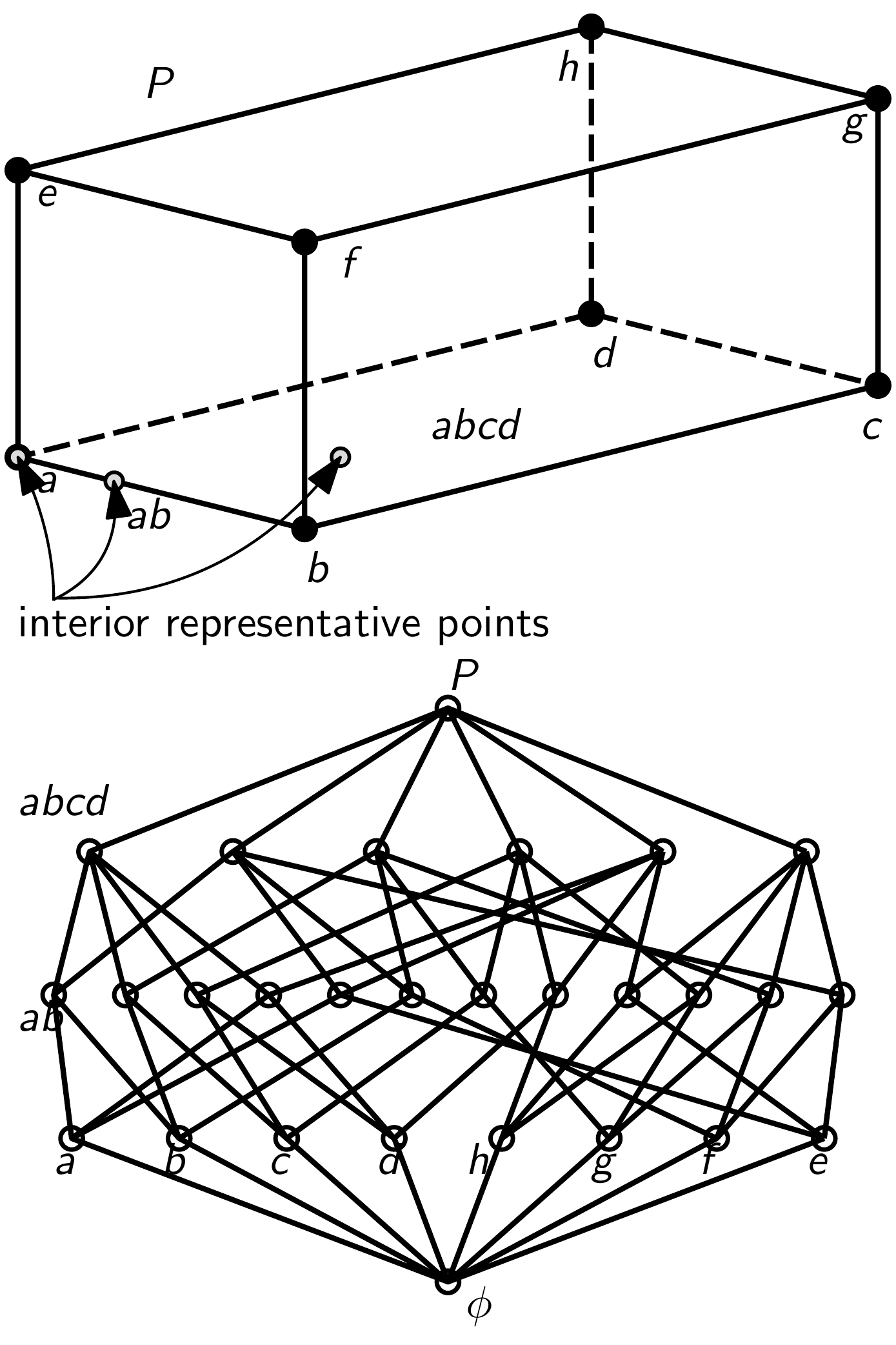} %
		\caption{Face lattice representation of a cuboid.}%
		\label{fig:facelattice}
		\vspace*{2ex}
	\end{minipage}
\end{figure*}

The concept of Minkowski sum is applicable to any geometric objects, convex or non-convex, discrete or continuous, finite or infinite. Fundamentally, the Minkowski sum of  geometric objects is an implicit  geometric object that corresponds to the Minkowski sum of the implicit set of points corresponding to the geometric objects. Formally the Minkowski sum of two set of points is defined as follows. 

\begin{definition}[Minkowski Sum] The \emph{Minkowski sum} of two sets of points $S$ and $S'$ is defined as the set of points $ \{ p + p'\ | \ \forall p \in S, \forall p' \in S'\}$. We denote the Minkowski sum of sets $S$ and $S'$ by the set $S \oplus S'$. 
	
The \emph{Minkowski} sum of two polytopes $P$ and $P'$ in $\Re^d$, denoted by $P \oplus P'$, is a polytope such that $ P\oplus P' = \{ p + p'\ | \ \forall p \in P, \forall p' \in P'\}$.
\end{definition}

The Minkowski sum operator $\oplus$ is associative and commutative and forms a group.
See Figure~\ref{fig:minkowski} for an illustration. We note that polytopes $P$ and $P'$ may have affine dimension less than $d$. In the definition above and elsewhere by sum of points we mean the vector sum  of the points when considered as vectors. We assume an origin of reference.

We give below some terminology and concepts about the polytopes and face lattice structure that we use in this paper.

The \emph{affine space} $A$ of a set of points $S$ in $\Re{^d}$ is the set of points such that they are  \emph{affine} combinations of the points in $S$, that is, $p \in A$ iff $ p = \sum_i \lambda_iq_i$ where $\forall i\ q_i \in S$, $\forall i\  \lambda_i \in \Re$ and $\sum_i \lambda_i = 1$. Every affine space has an associated \emph{affine dimension} which is cardinality of the basis. The \emph{dimension} of a face $f$ of a polytope $P$ is the dimension of the affine space of face $f$.

A hyperplane $h$  is a \emph{tangent} hyperplane to a polytope $P$ if $h$ touches $P$. The intersections $h \cap P$ of all possible tangent hyperplanes $h$ are the \emph{faces} of $P$, which are convex if $P$ is convex. Whenever we mention the \emph{dimension} of a face we mean the affine dimension of theface. If the dimension of a face is $k$ we term it as a $k$-face. A face $f$ is a \emph{subface} of $g$ if $f \subseteq g$. Also if a face $f$ is a subface of $g$ then $g$ is a \emph{superface} of $f$. 

The \emph{face lattice} representation of a polytope is an incidence graph (directed) which has nodes for each face of the polytope and there are two types of arcs from a face $f$, say of dimension $k$, of polytope, one type of arcs point to the immediate superfaces of $f$ of dimension $k+1$ and other type of arcs point to immediate subfaces of $f$ of dimension $k-1$. These are called \emph{incidences}. We also call the incident faces of dimension $k+1$ as \emph{immediate superfaces} and incident faces of dimension $k-1$ as \emph{immediate subfaces} of a face. In addition to obvious faces of the polytope, for the connected polytope of a singe interior, we also have a face of dimension $d$, the \emph{interior face}, and a face of dimension $-1$, the \emph{null face}. See Figure~\ref{fig:facelattice} for an example. 

In this paper, to simplify the notation, we use names like $k$-face and $k$-layer, where $k$ is a variable or number and not part of the name, thus we use the notation 0-face, 1-face, $(d-1)$-face, etc. At other places, however, we use the subscripts and superscripts as usual, such as $P_0$, $P_1$, $P_{d-1}$, etc.

For the purpose of the algorithm presented in paper, we augment the face lattice data structure such that each node representing a face contain an arbitrary strict relative interior point of the face (which is not on the boundary). By a relative interior point, we mean point in the $k$-dimensional interior for a $k$-face and by a strict interior point we mean that this interior point should not belong to any of the subfaces. For vertices, we keep the vertex point itself in the node. 

Let $P$ and $Q$ be the input convex polytopes in $\Re^d$ in this and successive sections. Let $P$ and $Q$ be given as face lattice structures. 

We make a specific assumption that for every  pair of pfaces $f \in P$ and $f' \in Q$, the dimension of $f\oplus f'$ is the sum of dimensions of faces $f$ and $f'$.  This assumption is to ensure that the facets of $P \oplus Q$ are formed by faces of dimensions $k$ and $d-k-1$ of $P$ and $Q$ respectively. We state this assumption as \emph{assumption of non-degeneracy} in the paper. We also assume that origin is in the interior of $P$ and $Q$, though this assumption is really not necessary.

We present an algorithm to compute Minkowski sum of polytopes in $\Re^d$. First we briefly present the well known algorithm to compute Minkowski sum of convex polygons in plane.

\section{Computing the Minkowski Sum of Convex Polygons in Plane}
\label{sec:planeminkowski}
\begin{figure*}[t]
	\begin{minipage}[b]{0.485\textwidth}
		\centering%
		\includegraphics[width=0.9\columnwidth]{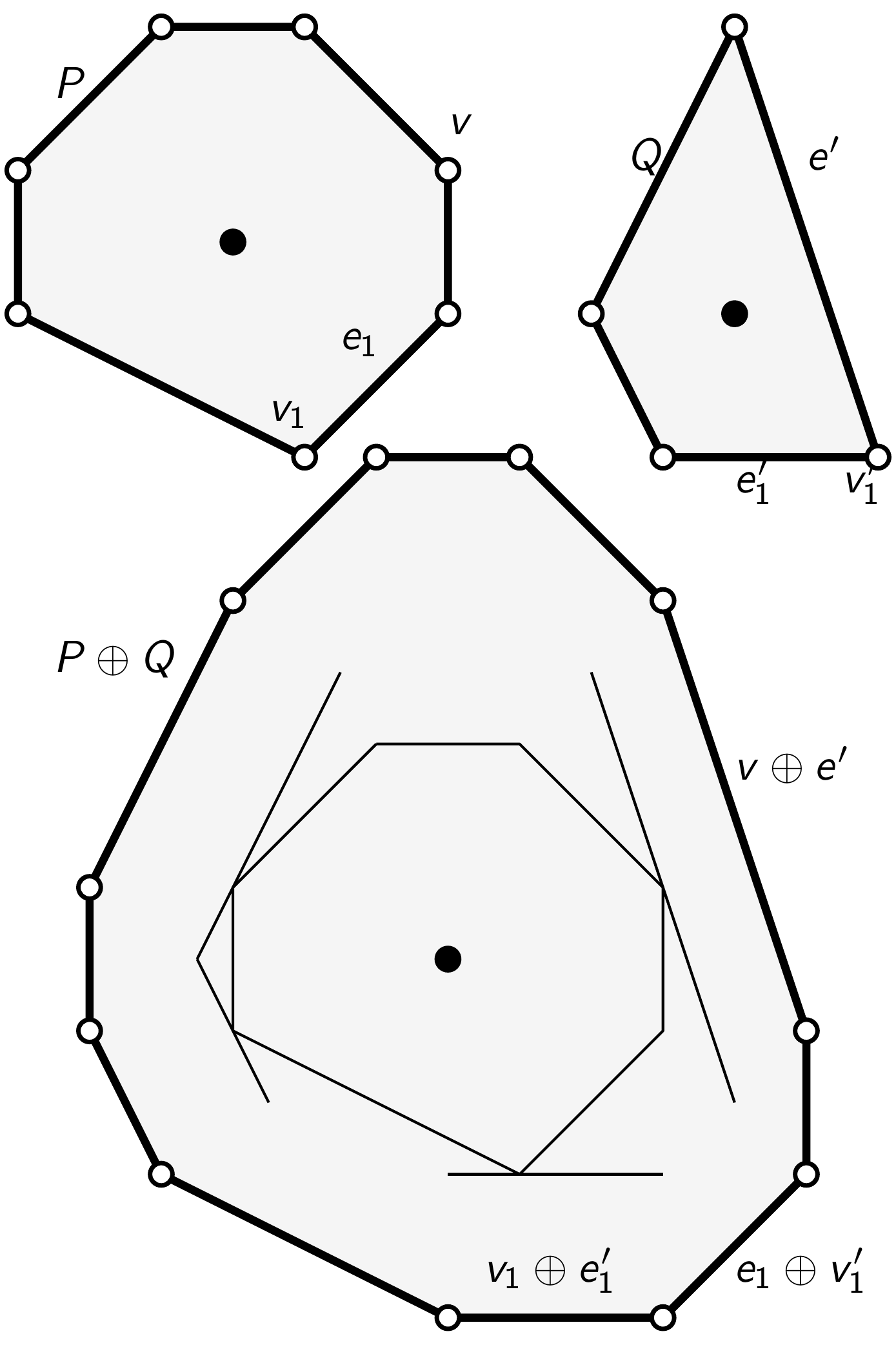}%
		\caption{Computing Minkowski sum of two convex polygons $P$ and $Q$ in $\Re^2$.%
		}%
		\label{fig:planeminkowski}%
		\vspace*{0ex}%
	\end{minipage}
	\hfill\quad%
	\begin{minipage}[b]{0.485\textwidth}
		\centering%
		\includegraphics[width=0.9\columnwidth]{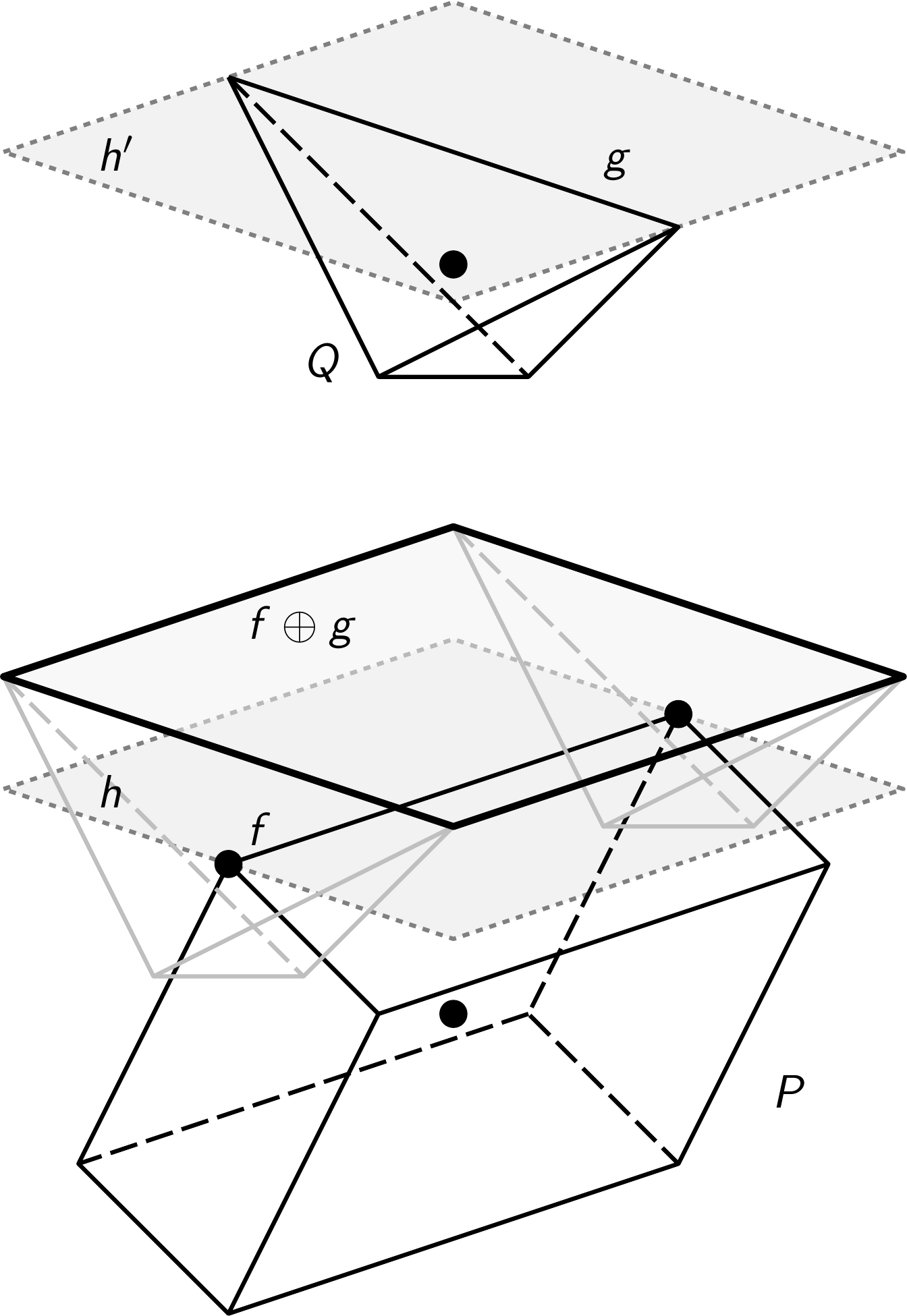}%
		\caption{ $f\oplus g$ is a face of $P \oplus Q$.}%
		\label{fig:lemmatangent}%
		\vspace*{2ex}%
	\end{minipage}
\end{figure*}
Let $P$ and $Q$ be two convex polygons in plane with number of vertices $n$ and $m$ respectively.  It is known that we can compute the Minkowski sum of convex polygons $P$ and $Q$ in time $O(n+m)$~\cite{BergOKO2008}.
We assume without loss of generality that both convex polygons are represented by a linked list of vertices in counter clockwise order.

First take an arbitrary initial edge $e$ of $Q$. By a linear scan of the vertices of $P$ in circular order we can compute the vertex $v_1=v$ of $P$ on which line parallel to $e'_1=e'$ is tangential. The edge $v_1 \oplus e'_1$ will be an edge of $P \oplus Q$. Let the vertices and edges of $P$ in counter clockwise direction from $v_1$ be $v_1$, $v_2$, $\dots$ and $e_1$, $e_2$, $\dots$ respectively. Let the edges and vertices of $Q$ in counter clockwise direction for $e'_1$ be $e'_1, e'_2, \dots$ and $v'_1$, $v'_2$, $\dots$ respectively.

Next we compute the vertex of $P$ starting from $v = v_i$ on which line parallel to next edge  of $e'=e'_k$ in $Q$, that is, $e'_{k+1}$, is tangential. Suppose, that vertex in $v_j$. Then $v_j\oplus e'_k$ will be an edge of $P\oplus Q$. If $v_i \neq v_j$ then $e_{i}\oplus v'_k$, $e_{i+1}\oplus v'_{k}$, $\dots$ $e_{i-1}\oplus v'_k$ are edges of $P\oplus Q$ in counter clockwise order from edge $v_i \oplus e'_k$ followed by  $v_j\oplus e'_{k+1}$. We reassign 
$v = v_j$ and $e' = e'_{k+1}$ and repeat until we come back to initial vertex $v_1$ in $P$ and $e'_1$ in $Q$. See Figure~\ref{fig:planeminkowski} for illustration.

Thus we have the following theorem.

\begin{theorem}
	The Minkowski sum of two convex polygons in plane $P$ and $Q$ of $n$ and $m$ vertices respectively can be computed in $O(n+m)$ time \cite{BergOKO2008}.
\end{theorem}

We can also compute the Minkowski sum of multiple convex polygons in linear time. But we cannot achieve linear time if we compute Minkowski sums of two polygons at a time. To compute the Minkowski sum of a multiple convex polygons in linear time we have to simultaneously traverse all the convex polygons together in the same counter-clockwise direction and compute the vertices and edges of the Minkowski sum in the sequence that they occur. Thus we have the following theorem for the Minkowski sum of a set of convex polygons.

\begin{theorem}
	The Minkowski sum of a set of convex polygons in plane can be computed in linear time of the input size.
\end{theorem}
\begin{proof} The proof follows from the simple construction given above.\qed
\end{proof}

\section{Computing the Minkowski Sum of two Convex Polytopes in higher dimensions}
\label{sec:genminkowski}

The algorithm in plane presented in Section~\ref{sec:planeminkowski} cannot be generalized to $\Re^d$ directly. One of the reasons is that there is no obvious traversal of polytopes in higher dimensions that is sequential and follows adjacencies and incidences, dissimilar to the case in plane. Since we are using usual representation of face lattice structures we cannot retrieve the adjacent faces of a face in $O(1)$ time.

As mentioned earlier, we call a $k$-dimensional face of a polytope is a $k$-face, i.e., 0-face, 1-face, $\dots$, $(d-1)$-face, etc. A 0-face is called a \emph{vertex}, a 1-face is called an \emph{edge}, a 2-face is called a  \emph{(polygonal) face} albeit ambiguously and a 3-face is called a  \emph{(polyhedral) cell}.
For a $k$-face $f$ or $k$-dimensional polytope $P$,  a $(k-1)$-face is called a \emph{facet}, a $(k-2)$-face is called a \emph{ridge} and a $(k-3)$-face is called a \emph{peak} of  the face $f$ or the  polytope $P$ respectively.

Let $P$ and $Q$ be two convex polytopes, such that the face lattice structures of $P$ and $Q$ are of size $n$ and $m$ respectively. For the convex polytope $P$, let $P_k$, $0\leq k < d$, be the set of $k$-faces of $P$, i.e., $P_0$, $P_1$, $\dots$, $P_{d-1}$. are the sets of 0-faces, 1-faces, $\dots$, $(d-1)$-faces of $P$ respectively. Similarly, $Q_0$, $Q_1$, $\dots$, $Q_{d-1}$ are the sets of 0-faces, 1-faces, $\dots$, $(d-1)$-faces of $Q$ respectively. 

Since face lattices are layered graph, we can label all the $k$-faces  as belonging to layer $k$, also named \emph{$k$-layer}, of the given face lattice of a polytope. All the immediate subfaces of layer $k$ will be in layer $k-1$ and all the immediate superfaces will be in layer $k+1$. Let the incidences of polytope $P$ at layer $k$ be denoted by $I_{P,k}$, which contain incidences to both superfaces as well as subfaces.

We show in this section how we can compute the Minkowski sum of convex polytopes $P$ and $Q$.
We do this in the following three stages.
\begin{description}
	\item[Stage 1:] We compute the $(d-1)$-faces, i.e., facets of $P \oplus Q$ which are of the type $f \oplus g$ corresponding to faces in $f\in P_0$ and $g \in Q_{d-1}$ . Optionally, we may also compute the  facets of $P \oplus Q$ which are of the type $f \oplus g$ corresponding to faces in $f\in P_{d-1}$ and $g \in Q_{0}$ .
	\item[Stage 2:] We compute rest of the facets of $P \oplus Q$  which are of the type $f \oplus g$ corresponding to faces in $f\in P_i$ and $g \in Q_{d-1-i}$ for $ 0 < i \leq d-1$.
	\item[Stage 3:] We compute all the $k$-faces of $P \oplus Q$ for $0 \leq k < d-1$ and  their incidences in face lattice of $P \oplus Q$.
\end{description}

\begin{figure*}[t]
	\begin{minipage}[b]{0.485\textwidth}
		\centering%
		\includegraphics[width=0.9\columnwidth]{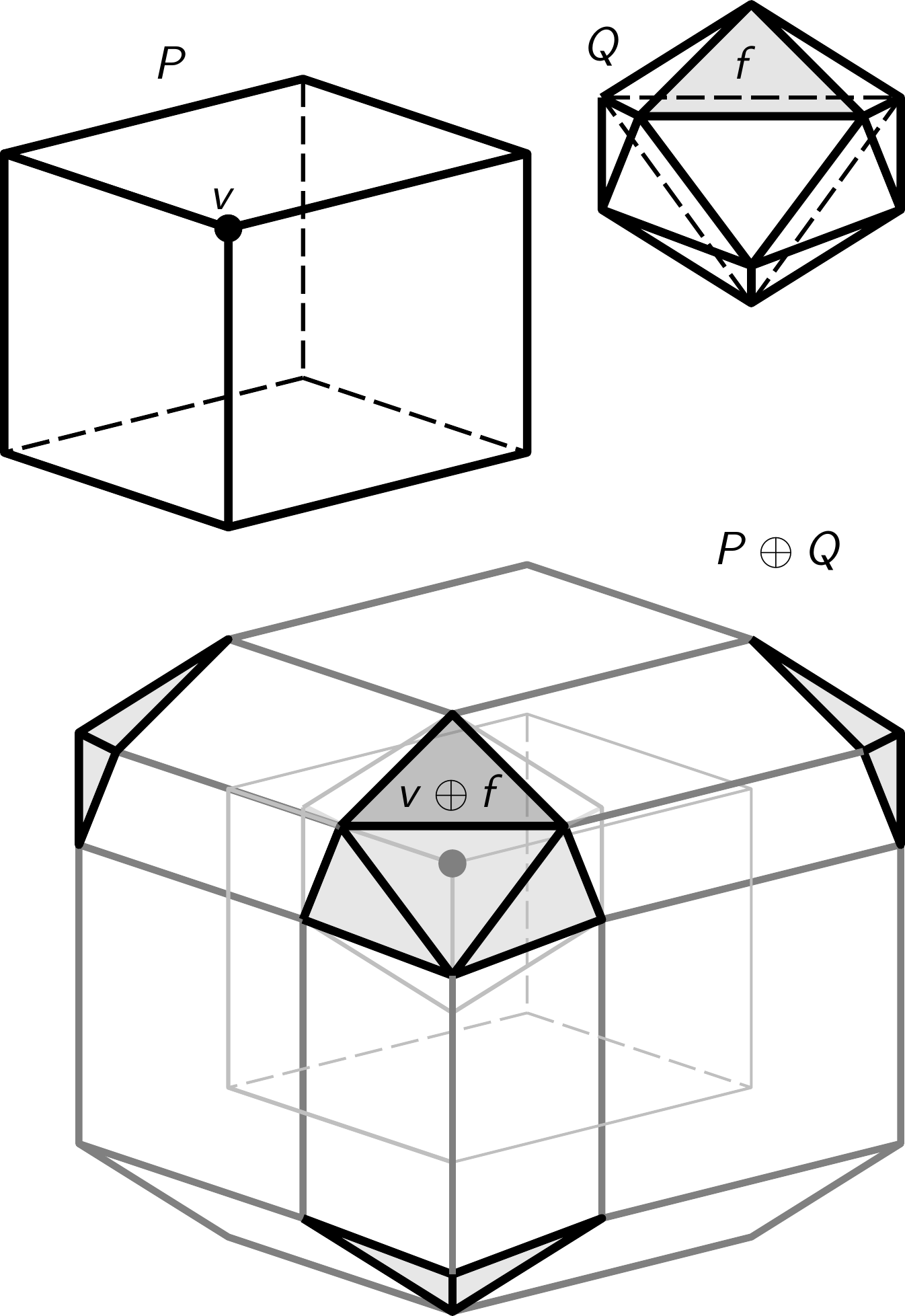}%
		\caption{Stage 1 of the algorithm in $\Re^3$.%
		}%
		\label{fig:stage1}%
		\vspace*{0ex}%
	\end{minipage}
	\hfill\quad%
	\begin{minipage}[b]{0.485\textwidth}
		\centering%
		\includegraphics[width=0.9\columnwidth]{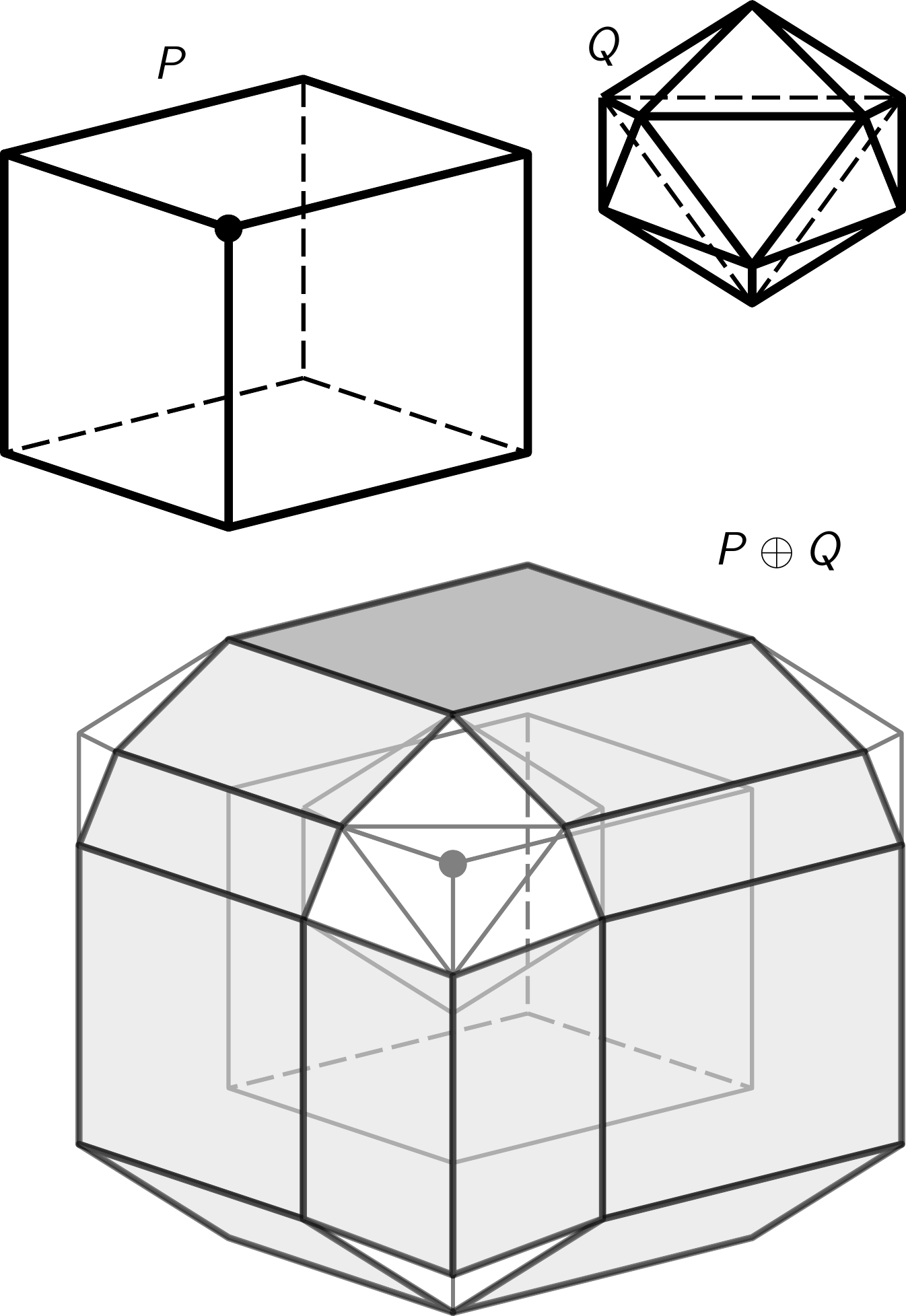}%
		\caption{Stage 2 of the algorithm in $\Re^3$.}%
		\label{fig:stage2}%
		\vspace*{0ex}%
	\end{minipage}
\end{figure*}

\subsection{Stage 1: Computing initial facets of $P\oplus Q$:}

In stage 1, we wish to compute an initial subset of the  $(d-1)$-faces of $P \oplus Q$ of particular kind mentioned below. We use the following lemma for that.

\begin{lemma}\label{lem:tangent}
If parallel tangent hyperplanes $h$ and $h'$  touch convex polytopes $P$ and $Q$, respectively, on the same side  at  faces $f$ and $g$, respectively, not necessarily of same dimensions, then $f \oplus g$ is a face of $P \oplus Q$.
\end{lemma}
\begin{proof}
	Consider $h \oplus h'$ which will be tangent hyperplane of $P\oplus Q$ on the same side as $h$ and $h'$ on $P$ and $Q$ respectively. It will be tangent because $p \oplus q$, for every $p\in P$, $q\in Q$, will be on the same side. Hence $f\oplus g$ will be on the boundary of $P \oplus Q$. It can also be seen that if $h$ and $h'$ are rotated infinitesimally such that they do not touch $P$ and $Q$, respectively, anywhere else other than $f$ and $g$, respectively, then $h\oplus h'$ also does not touch $P\oplus Q$ anywhere else other than $f\oplus g$. See Figure~\ref{fig:lemmatangent}. This is true because of convexity of $P$.\qed
\end{proof}

We note that this lemma is also used elsewhere in the paper.

We  compute initial facets of polytope $P \oplus Q$ with the help of lemma above as follows. For every facet $f$ in $Q_{d-1}$, we compute the vertex $v$ in $P_0$ such that the hyperplane $h$ parallel to facet $f$ of $Q$ touches $P$ at $v$ and $P$ is on the same side of $h$ as $Q$ is of $f$. By Lemma~\ref{lem:tangent} and taking $h'$ to be hyperplane supporting $f$, we conclude that $v \oplus f$ is a $(d-1)$-face of $P \oplus Q$. 

The above computation can be done in $O(n)$ as follows. First we compute the list of vertices $P_0$ of polytope $P$ from the face lattice. This needs to be done only once for every facet $f$ of $Q$. Next, for every facet $f \in Q_{d-1}$, we compute the outward normal $\hat n$. We need to select the vertex $v$ of $P$ such that $P$ is on the same side of parallel hyperplane $h$ touching $v$ with direction $\hat n$. This can be done  by  computing $\hat n \cdot u$ for every vertex $u \in P_0$ and selecting the vertex $u$ for which this is maximum. This vertex is the required vertex $v$. We add facet $v \oplus f$ for facet $f$ and repeat.

Both $P_0$ and $Q_{d-1}$ can be scanned by any  traversal that takes linear time of two layers of $P$ (0-layer and 1-layer) and $Q$ ($(d-1)$-layer and $(d-2)$-layer). See Figure~\ref{fig:stage1}.
The complexity of stage 1 is given by the following lemma.

\begin{lemma}
	\label{lem:stage1}
	Stage 1 of the algorithm can be completed in $O(nm)$ time.
\end{lemma}
\begin{proof}
	We can compute the list of vertices in $P$ and list of facets in $Q$ in time $O(|P_0|+|Q_{d-1}|)$ (these are just incident of $(-1)$-face of $P$ and $d$-face of $Q$).	
	Then for every facet $f \in Q_{d-1}$ we can compute the vertex $v \in P_0$ such that $v\oplus f$ is a facet of $P\oplus Q$ in time $O(|P_0|\cdot|Q_{d-1}|)$. Thus total time is $O(nm)$.\qed
\end{proof}

As mentioned earlier, we can optionally switch the role of $P$ and $Q$ and compute the faces of type $g\oplus u$ of the Minkowski sum $P\oplus Q$, where $g\in P_{d-1}$ and $u\in Q_0$.

\subsection{Stage 2: Computing rest of the facets of $P\oplus Q$:}

In stage 2, we compute all the remaining facets of $P\oplus Q$. These are the facets of type $f \oplus g$, where $f$'s are non-vertex faces of the convex polytope $P$ and $g$'s are non-facet faces of the convex polytope $Q$.

We traverse the face lattice structure layer by layer, in this stage, starting from $0$-layer of $P$ and moving to higher layers, and from $(d-1)$-layer of $Q$ and moving to lower layers. Assume we have completed processing of $(k-1)$-layer of $P$, i.e., $P_{k-1}$, and correspondingly $(d-k)$-layer of $Q$, i.e., $Q_{d-k}$. We wish to compute facets of the type $f\oplus g$ where $f \in P_{k}$ and $g \in Q_{d-k-1}$. We compute these as follows.

For every face $f \in P_{k}$ and $g \in Q_{d-k-1}$ we compute the hyperplane $h$ supporting $f\oplus g$. Due to assumption of non-degeneracy the dimension of $f\oplus g$ will be $d-1$. Next we compute hyperplanes $h'$ and $h''$ parallel to $h$ and passing through $f$ and $g$ respectively. We need to check if $h'$ is tangent hyperplanes to $P$ with $P$ on the same side as $h$; and also if $h''$ is tangent hyperplanes to $Q$ with $Q$ on the same side  as $h$. We do this by using the interior point stored in the augmented face lattice structure of the immediate superfaces of $f$ and $g$ respectively. If all the interior points of the immediate superfaces are on the same side as origin (note that we have assumed that the origin lies inside $P$ and $Q$, otherwise, we need to take the interior point of $P_d$ and $Q_d$ respectively), then $f \oplus g$ is a facet of $P\oplus Q$.

Since $P$ and $Q$ are convex polytopes if all the interior points of immediate superfaces of $f \oplus g$ are on the same side of the supporting hyperplane, then whole of $P$ or $Q$ is on the same side. Thus the supporting hyperplane will be a tangent, the conditions of Lemma~\ref{fig:lemmatangent} apply and $f \oplus g$ will be a facet of $P \oplus Q$.

\begin{lemma}
	\label{lem:stage2}
		Stage 2 of the algorithm can be completed in $O(nm)$ time.
\end{lemma}
\begin{proof}
	For each face $g$ in $Q_{d-k-1}$, the time taken is $O(|I_{P,k}|)$ for all the faces $f \in P_k$. Thus total time taken will be $O(|I_{P,k}|\cdot |I_{Q,d-1-k}|)=O(nm)$ for all the faces of $Q_{d-1-k}$ that we consider in stage 2 of the algorithm. We scan every incidence in face lattice of $P$ and  $Q$ at most once for each pair of faces $f\in P$ and $g\in Q$.\qed
\end{proof}

We note, that the method mentioned in the stage 2 is equally applicable for computing of the facets that were computed in stage 1 with a higher time complexity.

\begin{figure*}[t]
	\begin{minipage}[b]{0.485\textwidth}
		\centering%
		\includegraphics[width=0.9\columnwidth]{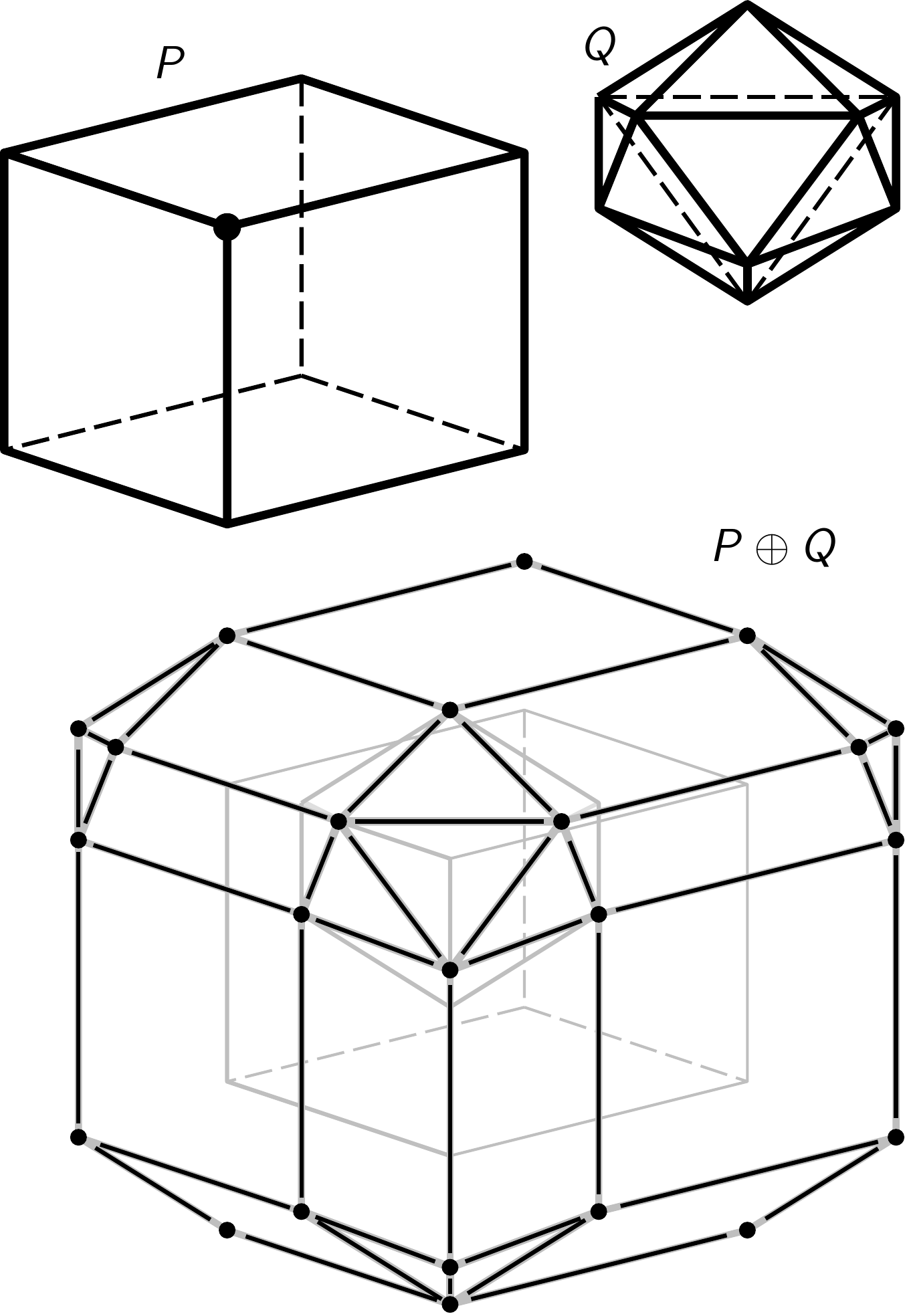}%
		\caption{Stage 3 of the algorithm in $\Re^3$.%
		}%
		\label{fig:stage3}%
		\vspace*{0ex}%
	\end{minipage}
	\hfill\quad%
	\begin{minipage}[b]{0.485\textwidth}
		\centering%
		\includegraphics[width=0.9\columnwidth]{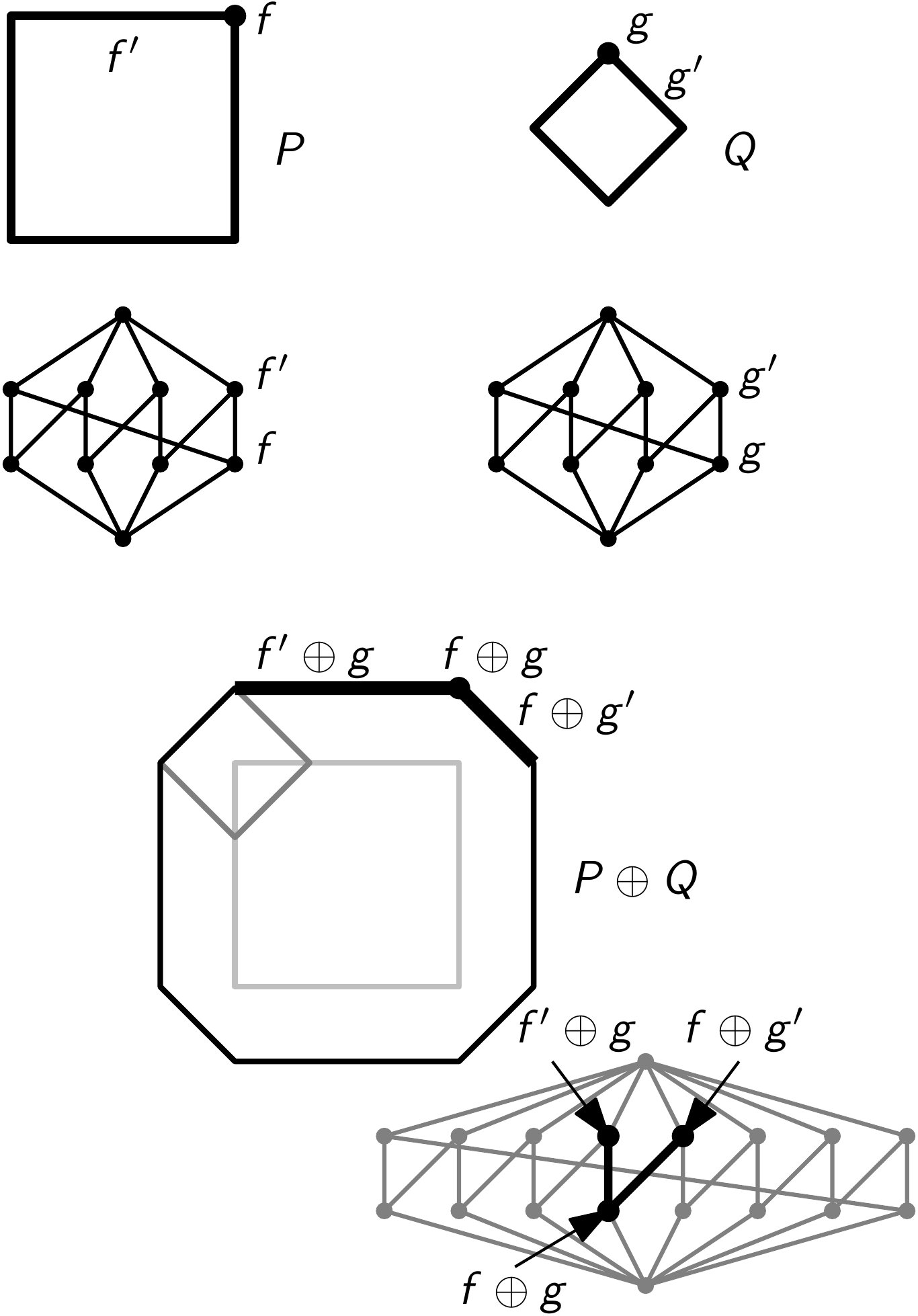}%
		\caption{Face lattice of $P\oplus Q$ in the algorithm in $\Re^2$.}%
		\label{fig:final}%
		\vspace*{0ex}%
	\end{minipage}
\end{figure*}

\subsection{Stage 3: Computing all the non-facet subfaces of $P\oplus Q$:}

In stage 3, we compute all the non-facet faces of $P\oplus Q$ and their correct incidences. We use the following lemmas and corollary for this.

\begin{lemma}
\label{lem:subface}
Let $P$ and $Q$ be convex polytopes.
For any $f,f' \in P$ and $g \in Q$, such that $f'$ is a  subface of $f$ in $P$, if $f \oplus g$ is a $k$-face of $P \oplus Q$, $0 \leq k \leq d-1$ then $f'\oplus g$ is a subface of $f\oplus g$ in $P \oplus Q$.
\end{lemma}
\begin{proof}
Proof follows from the definition of Minkowski sum.\qed
\end{proof}
\begin{corollary}
For any $f \in P$ and $g,g' \in Q$, such that $g'$ is a subface of $g$ in $Q$, if $f \oplus g$ is a face of $P \oplus Q$ then $f\oplus g'$ is a  subface of $f\oplus g$ in $P \oplus Q$.
\end{corollary}
\begin{proof}
	This is because Minkowski sum is commutative.\qed
\end{proof}

The above lemmas provide a sufficient condition for subfaces. Below we provide a necessary condition.

\begin{lemma}
\label{lem:superface}
Let $P$ and $Q$ be convex polytopes.
For any $f\in P$ and $g \in Q$, if $f \oplus g$ is a $k$-face of $P \oplus Q$, $0\leq k < d-1$, then there exists a superface $f'$ of $f$ in $P$ or a superface $g'$ of $g$ in $Q$ such that  $f'\oplus g$ or  $f \oplus g'$ is a superface of $f\oplus g$ in $P \oplus Q$ and is of dimension $k+1$.
\end{lemma}

Lemma~\ref{lem:subface} and Lemma~\ref{lem:superface} provides us with a method to compute all the lower dimensional subfaces of $P \oplus Q$. This can be done by carefully duplicating the incidence relationships in $P$ and $Q$ for the relevant facets of $P \oplus Q$.  See Figures~\ref{fig:stage3} and~\ref{fig:final}.

\begin{figure*}[t]
	\begin{minipage}[b]{0.485\textwidth}
		\centering%
		\includegraphics[width=0.9\columnwidth]{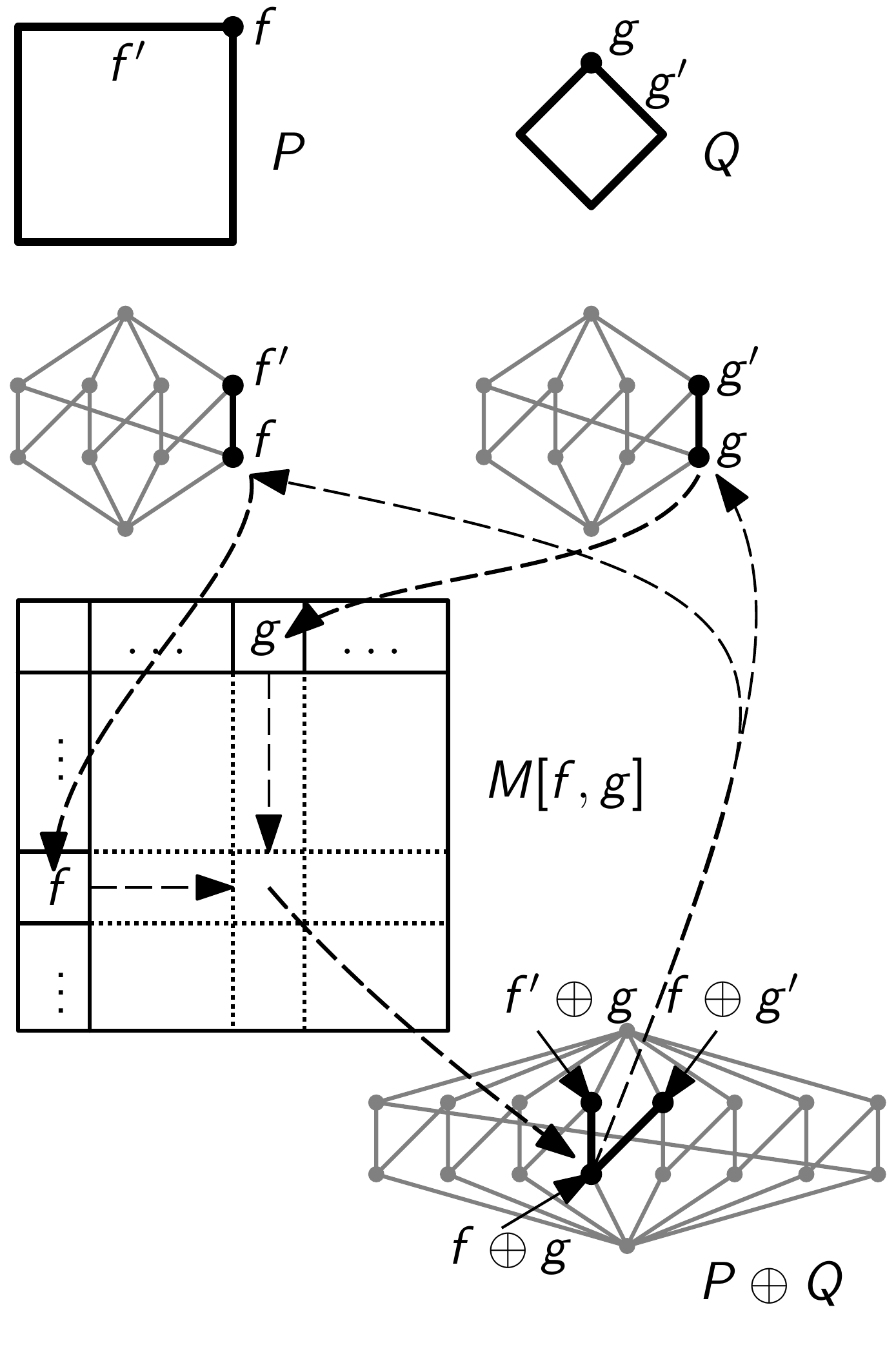}%
		\caption{Creating incidences in Stage 3 of the algorithm using table $M$.%
		}%
		\label{fig:table}%
		\vspace*{0ex}%
	\end{minipage}
	\hfill\quad%
	\begin{minipage}[b]{0.485\textwidth}
		\centering%
		\includegraphics[width=0.9\columnwidth]{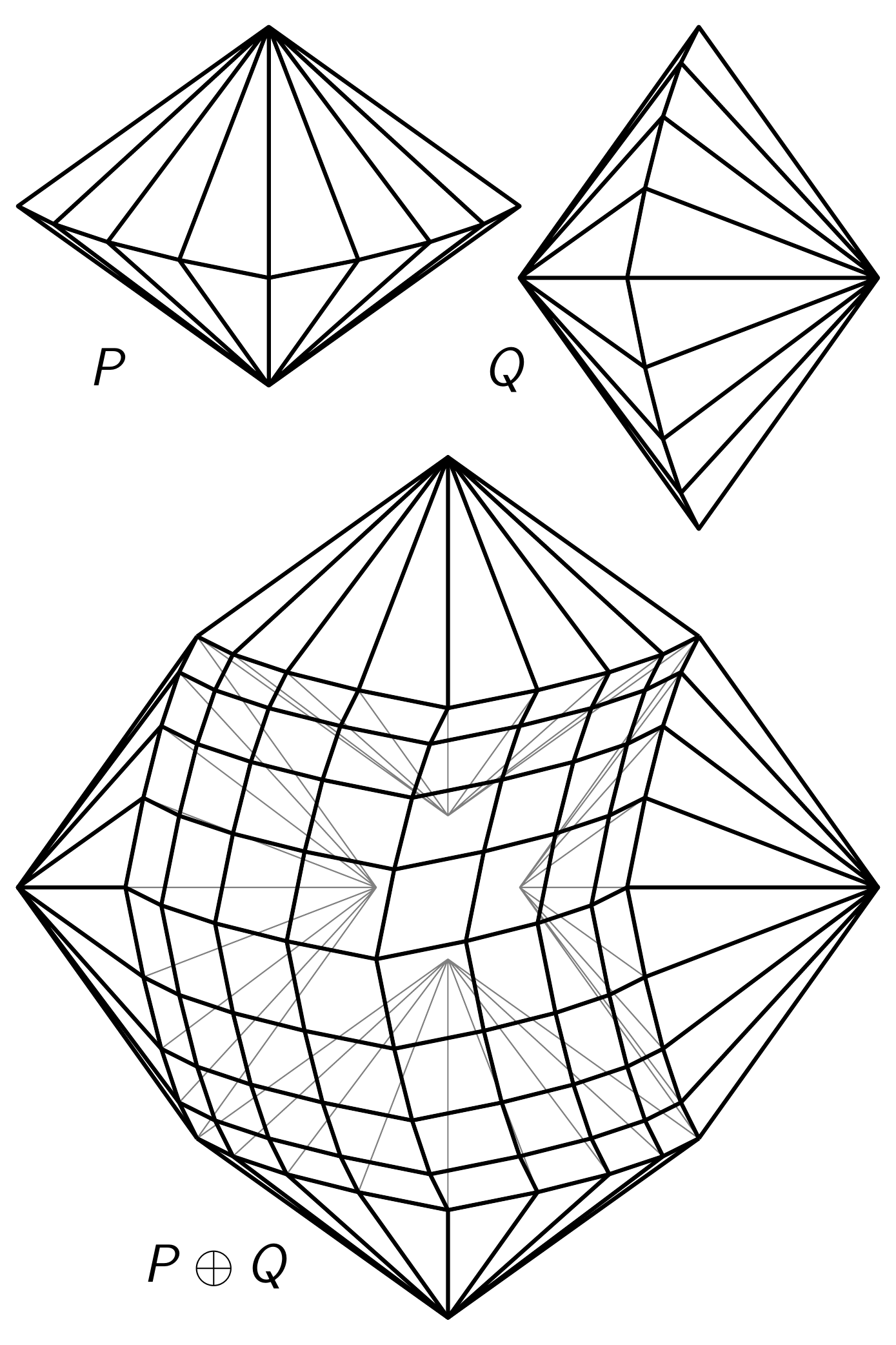}%
		\caption{Worst case: Even the vertices, even in $\Re^3$ can be $O(nm)$.}%
		\label{fig:worst}%
		\vspace*{0ex}%
	\end{minipage}
\end{figure*}

We describe the algorithm to do the above in detail as follows. We give unique indices to each face of $P$ and store the indices in the nodes of each face of $P$ of each dimension 0 to $d-1$. There will be $O(n)$ indices for the faces of $P$. We do the same for faces of $Q$. There will be $O(m)$ indices for the faces of $Q$. 
Next we create a two dimensional table, say $M$, of maximum size $O(nm)$, such that $M[f,g]$ will keep a link to face $f \oplus g$ of convex polytope $P\oplus Q$, if it exists. In the beginning, we initialise the table $M$ with empty links. Then, whenever, we determine if $f \oplus g$ is a face of $P\oplus Q$, for $f\in P$ and $g\in Q$, we keep the link to the node of the face $f \oplus g$ in the face lattice assigned to the cell $M[f,g]$. 
We use table $M$ to create nodes in the face lattice of $P\oplus Q$ only once for each face.

We illustrate the process below by giving an example. Suppose $f' \oplus g$ and $f \oplus g'$ are $(k+1)$-faces of $P\oplus Q$, and $f\oplus g$ is a $k$-face of $P\oplus Q$ such that $f$ is a facet of $f'$, $f,f' \in P$ and $g, g' \in Q$, $g$ is a facet of $g'$. Then first time we encounter 
$f\oplus g$ from $f' \oplus g$, we find the table entry $M[f,g]$ empty, we create the node for $f\oplus g$, link it with table entry, and assign the incidence in the output face lattice for $f' \oplus g$. Next time, when we encounter $f\oplus g$ from $f \oplus g'$, we already have a link in the table entry for $M[f,g]$ to the node for $f\oplus g$. We use this link to create  the correct incidences in the output face lattice for $f \oplus g'$ and stop further subface processing, as we have already processed the subfaces of $f\oplus g$ earlier. See Figure~\ref{fig:table}.

Thus we have the following lemma.

\begin{lemma}
	\label{lem:stage3}
Stage 3 of the algorithm can be completed in $O(nm)$ time.
\end{lemma}
\begin{proof}
The proof follows from the discussion above and the size of the table $M$ is maximum $O(nm)$.\qed
\end{proof}

Combining the time complexities of the three stages, we have the following theorem.
\begin{theorem}
	The Minkowski sum of two convex polytopes $P$ and $Q$  in $\Re^d$ of face-lattice structure sizes $n$ and $m$, respectively, can be computed in $O(nm)$ time.
\end{theorem}
\begin{proof}
	The proof follows from the Lemmas~\ref{lem:stage1},~\ref{lem:stage2} and~\ref{lem:stage3}.\qed
\end{proof}

\section{Computing the Minkowski Sum of Multiple Convex Polytopes in $\Re^d$}
\label{sec:multiple}

Let $P_1$, $P_2$, $\dots$, $P_n$ be $n$ convex polytopes of face lattice sizes $N_1$, $N_2$, $\dots$, $N_n$, respectively. We can compute the Minkowski sum $P_1 \oplus P_2 \oplus \cdots \oplus P_n$ in worst case time complexity $O(N_1N_2 + N_1N_2N_3+\cdots+N_1N_2\dots N_n) = O(nN_1N_2\dots N_n)$ by the method presented in the Section~\ref{sec:genminkowski}. However, we can do better if we compute the Minkowski sum of $n$ convex polytopes simultaneously.

The algorithm for Minkowski sum of multiple polytopes in $\Re^d$ will constitute of three stages as before.

In stage 1, we fix facets, say $f$, one at a time, of convex polytope $P_i$ and find vertices $v_j$'s of  convex polytopes $P_j$'s, $j \neq i$, such that condition of tangent hyperplanes parallel to hyperplane supporting $f$ in Lemma~\ref{lem:tangent} is satisfied. The face  $v_1 \oplus v_2 \oplus \cdots \oplus v_{i-1} \oplus f \oplus v_{i+1} \oplus \cdots v_n$ will  be a facet of $P_1 \oplus P_2 \oplus \cdots \oplus P_n$.

In stage 2, we take faces, say $f_i$'s $\in P_i$'s, such that total sum of dimensions of $f_i$'s is $d-1$. Then we check if $f_1 \oplus f_2 \oplus \cdots \oplus f_n$ is a facet of $P_1 \oplus P_2 \oplus \cdots \oplus P_n$ by checking if the strict interior points of all the immediate superfaces are on the same side as the origin (which is again assumed in the interior of all polytopes, otherwise we take the interior point of polytopes) for parallel hyperplanes to $f_1 \oplus f_2 \oplus \cdots \oplus f_n$ touching $P_i$'s.

Stage 3 of the algorithm is exactly same due to  the following generalization of Lemmas~\ref{lem:subface} and~\ref{lem:superface}.

\begin{lemma}
Let $\forall i$, $1 \leq i \leq n$, $f_i \in P_i$, and  $g_i$ be a superface of $f_i$. If 
$f_1 \oplus f_2 \oplus \cdots \oplus f_{j-1} \oplus g_j \oplus f_{j+1}\oplus \cdots f_n$ is a face of $P_1 \oplus P_2 \oplus \cdots \oplus P_n$ then  $f_1 \oplus f_2 \oplus \cdots \oplus f_n$ is a subface of $f_1 \oplus f_2 \oplus \cdots \oplus f_{j-1} \oplus g_j \oplus f_{j+1}\oplus \cdots f_n$ in $P_1 \oplus P_2 \oplus \cdots \oplus P_n$. 

Also if  $f_1 \oplus f_2 \oplus \cdots \oplus f_n$ is a $k$-face of $P_1 \oplus P_2 \oplus \cdots \oplus P_n$, $0\leq k < d-1$ then there exists a superface $g_j$ of $f_j$ for some $j$ such that  $f_1 \oplus f_2 \oplus \cdots \oplus f_n$ is a subface of $f_1 \oplus f_2 \oplus \cdots \oplus f_{j-1} \oplus g_j \oplus f_{j+1}\oplus \cdots f_n$ in $P_1 \oplus P_2 \oplus \cdots \oplus P_n$ for that $j$.
\end{lemma}

Thus we have the following theorem.

\begin{theorem}
	The Minkowski sum $P_1 \oplus P_2 \oplus \cdots \oplus P_n$ of $n$ convex polytopes $P_1$, $P_2$, $\dots$, $P_n$  of face lattice sizes $N_1$, $N_2$, $\dots$, $N_n$, respectively can be computed in time $O(\prod_{i=1}^n N_i)$.
\end{theorem}

\section{Concluding Remarks}
\label{sec:conclusion}

In the paper we presented an algorithm to compute the Minkowski sum that has asymptotic running time as product of input sizes of the convex polytopes. However, since in the worst case,  for $d\geq 3$, in the Minkowski sum of two and three polytopes, the number of vertices of the resulting polytope is known to be the product of numbers of input vertices, therefore, the algorithm remains linear in the output size in the worst case~\cite{KaravelasMT2012}.

We have assumed non-degeneracy of a special type in our paper, but we can easily deal with degeneracies for the algorithm described in this paper as the lemma mentioned in the paper essentially hold except for the requirement of dimensions.

\bibliographystyle{plain}  
\bibliography{myrefs}    

\end{document}